\documentclass[subscriptcorrection,upint,varvw,mathalfa=cal=euler,balance,hyphenate,nolists]{asmejour}

\usepackage{bm}
\usepackage{graphicx}

\newtheorem{proposition}{Proposition}
\newtheorem{theorem}{Theorem}
\newtheorem{mydef}{Definition}
\newtheorem{remark}{Remark}

\makeatletter
\newenvironment{proof}[1][Proof]{\par
  \normalfont \topsep6pt\labelsep6pt
  \trivlist
  \item[\hskip\labelsep\itshape #1.]\ignorespaces
}{%
  \hfill\ensuremath{\hfill\ensuremath{\square}} 
  \endtrivlist\@endpefalse
}
\makeatother

\newcommand{\lam}{\boldsymbol{\lambda}}
\newcommand{\lamdes}{\boldsymbol{\lambda}_{\mathrm{des}}}
\newcommand{\lamtil}{\tilde{\boldsymbol{\lambda}}}
\newcommand{\q}{\boldsymbol{q}}
\newcommand{\qdes}{\boldsymbol{q}_{\mathrm{des}}}
\newcommand{\qtil}{\tilde{\boldsymbol{q}}}
\newcommand{\s}{\boldsymbol{s}}
\newcommand{\R}{\mathbb{R}}
\newcommand{\Sph}{\mathbb{S}^2}
\newcommand{\norm}[1]{\left\lVert #1 \right\rVert}

\newcommand{\betat}{\tilde{\boldsymbol{\beta}}}

\newcommand{\tp}{^\top}

\hypersetup{%
	pdfauthor={Connor Calme, Lundon Salley, Luis F. Zapata-Rivera, Aldo J. Munoz-Vazquez},                       		   
	pdftitle={Robust PIDNet Control of a Dual-Actuator Thrust Vectoring Platform},                  	
	pdfkeywords={Thrust Vector Control, Nonlinear PID, Neural Network Control, Radial Basis Function, Lyapunov Stability},
	pdfsubject = {Robust controller for attitude tracking on a dual-actuator thrust vectoring platform},			
	pdflicenseurl={https://ctan.org/pkg/asmejour},
}

\JourName{Computational and\\ Nonlinear Dynamics}

\begin{document}

\SetAuthorBlock{Connor Calme}{Department of Aerospace Engineering,\\
    Embry-Riddle Aeronautical University,\\
   3700 Willow Creek,\\
   Prescott, Arizona, USA \\
   email: calmec@my.erau.edu} 

\SetAuthorBlock{Lundon Salley}{Department of Aerospace Engineering,\\
    Embry-Riddle Aeronautical University,\\
   3700 Willow Creek,\\
   Prescott, Arizona, USA \\
   email: salleyl@my.erau.edu} 

\SetAuthorBlock{Luis F. Zapata-Rivera}{Department of Computer, Electrical \& Software Engineering,\\
    Embry-Riddle Aeronautical University,\\
   3700 Willow Creek,\\
   Prescott, Arizona, USA \\
   email: zapatarl@erau.edu} 

\SetAuthorBlock{Aldo J. Munoz-Vazquez\CorrespondingAuthor}{Department of Computer, Electrical \& Software Engineering,\\
Embry-Riddle Aeronautical University,\\
   3700 Willow Creek,\\
   Prescott, Arizona, USA \\
   email: munozvaa@erau.edu} 

\title{Robust PIDNet Control of a Dual-Actuator Thrust Vectoring Platform}


\begin{abstract}
This paper presents a robust controller for attitude tracking
on a dual-actuator thrust vectoring platform.
The control objective is to track a desired thrust direction
on the unit sphere using two linear actuators connected
through a universal joint.
The proposed framework combines a bounded nonlinear PD
action with a radial basis function (RBF) network that
generates state-dependent integral compensation through
online weight adaptation, relying solely on tracking error
and its time derivative from actuator-displacement measurements.
The control law is formulated as a coupled
multiple-input--multiple-output structure, enabling the RBF
network to compensate cross-channel effects,
including direction-dependent friction and geometric
coupling, with model-free implementation.
Stability is established through a convex Lyapunov function
whose gradient is naturally bounded along the error manifold.
An $\mathcal{H}_\infty$ gain bound is derived for the extended error,
with a local linearized interpretation in terms of tracking accuracy.
The proposed controller is validated experimentally on a
two-degree-of-freedom gimbal rig,
achieving improved performance relative to a conventional PID and
to a super-twisting controller, with online
adaptation reducing the ITNE by $10\%$ over the
non-adaptive baseline at no additional control effort.
\end{abstract}

\date{Version \versionno, \today}

\maketitle 


\section*{Acknowledgment} 

The authors sincerely thank the Rocket Development Lab at Embry-Riddle Aeronautical University (Prescott campus) for designing and building the thrust vector control (TVC) platform and for allowing its use during summer research activities.

\section*{Funding Data}

\begin{itemize}
\item The authors acknowledge support from Embry-Riddle Aeronautical University through the Undergraduate Research Institute and its Summer Initiative Award program.
\item The authors acknowledge the Department of Computer, Electrical, \& Software Engineering at Embry-Riddle Aeronautical University for providing access to laboratory facilities and for technical support.
\end{itemize}

\section{Introduction}

Thrust vector control (TVC) has proven to be a reliable 
alternative for steering a launch vehicle in cases where
control surfaces are unreliable or unavailable ~\cite{sutton2017}. 
Modern configurations usually employ gimbaled propulsion systems
powered by hydraulic or electromechanical
actuators~\cite{lazic2007,ensworth2013}. Classical
approaches mainly rely on linear quadratic optimization
and conventional PID
controllers~\cite{sopegno2022lqr,ahmad2021}. However,
the presence of strong nonlinear coupling and external
disturbances crave for robust techniques.

Sliding mode control has been applied to TVC at the
vehicle level~\cite{utkin1992,yeh2013,wang2018}, proving high
accuracy and robust performance, but at the expense of
high-frequency oscillations in the control signal,
limiting its practicality in actuator-driven hardware.
This has motivated smooth alternatives based on neural
networks or fuzzy approximators \cite{shi2024new,garcia2024cnn}, where unmodeled
dynamics are compensated online without requiring
explicit disturbance bounds.

Classical adaptive neural network controllers for Euler-Lagrange systems
combine Lyapunov-derived weight update laws, relying on the sigma modification
to guarantee uniform ultimate boundedness~\cite{lewis1998}. 
A parallel line of research,
originated by Slotine and Li~\cite{slotine1987,slotine1988}
and refined by Lor\'{\i}a and Panteley~\cite{loria2005} 
achieves stronger convergence guarantees by exploiting a
linearly parameterized regressor
$Y(q,\dot{q},\dot{q}_r,\ddot{q}_r)\,\theta$, which is constructed
from the complete symbolic expressions of the inertia,
Coriolis, and gravitational matrices, with only a
constant parameter vector $\theta$ treated as unknown.
This structural requirement restricts the class of
compensable uncertainty to parametric variation in an
otherwise known model; unstructured nonlinearities such
as stiction, backlash, and direction-dependent friction,
which dominate actuator-level dynamics in
electromechanical TVC hardware, fall outside the
regressor span and cannot be addressed within that
framework. Moreover, the resulting control signals grow
without bound as the tracking error increases, posing a
practical risk of actuator saturation that the theory
does not accommodate.
More recent deep neural network
extensions~\cite{patil2022,patil2022resnet}
remain limited to simulation studies.
In the TVC domain, neural networks have been applied to
attitude control~\cite{garcia2024cnn}, fault
tolerant control~\cite{yikilmaz2022deep}, and actuator friction
compensation~\cite{ruan2021friction}; of these, only~\cite{garcia2024cnn}
includes experimental results, while the remaining studies
are limited to simulation.

This paper proposes a model-free controller that
circumvents these limitations: control signals are
bounded by construction through a nonlinear
proportional--derivative (PD) action, and unstructured
nonlinearities are compensated online by a radial basis
function (RBF) network with sigma-modified weight
adaptation, without requiring any plant model.

The present work makes the following contributions:
\begin{enumerate}
\item A coupled MIMO PIDNet controller for a parallel
thrust vectoring mechanism, combining a smooth and bounded
nonlinear PD term with supervised RBF weight adaptation
via sigma-modification, designed directly in actuator
space and relying only on actuator-displacement
measurements.
\item A stability proof using a convex Lyapunov candidate
function that yields an explicit $\mathcal{H}_\infty$
gain from the RBF approximation error to the saturated
extended error, with a local linearized extension to the
actuator tracking error.
\item Experimental validation on a two-degree-of-freedom
electromechanical TVC rig, demonstrating improved
tracking over PID and super-twisting baselines, with
measurable gain from online adaptation.
\end{enumerate}

\section{Platform Geometry and Dynamics}

\subsection{Geometry in Thrust Direction Space}

The platform constitutes a two-degree-of-freedom parallel
mechanism: two linear actuators, which are mounted on a fixed base,
drive a common top through a universal joint,
forming a closed kinematic chain. Because both actuators
act on the same rigid body, the forward kinematics are
inherently coupled, as a consequence, the resulting cross-axis dynamics
depend on friction, backlash, and geometric nonlinearities
that resist tractable closed-form modeling. This motivates
the model-free approach developed in Section~\ref{sec:controller}.

The platform consists of a fixed base and a tilting
top mount connected by a rigid bar through a universal
joint at height~$h_1$ above the base. The bar has
length~$h_2$, and the joint permits rotation about two
orthogonal axes while blocking rotation about the bar
axis. Two linear actuators connect pins on the base
(at radius~$r_b$ along the $x$- and $y$-axes,
respectively) to corresponding pins on the top mount
(at radius~$r_t$), see Fig.~\ref{fig:platform}.

\begin{figure}[t]
  \centering
  \includegraphics[width=0.7\columnwidth]{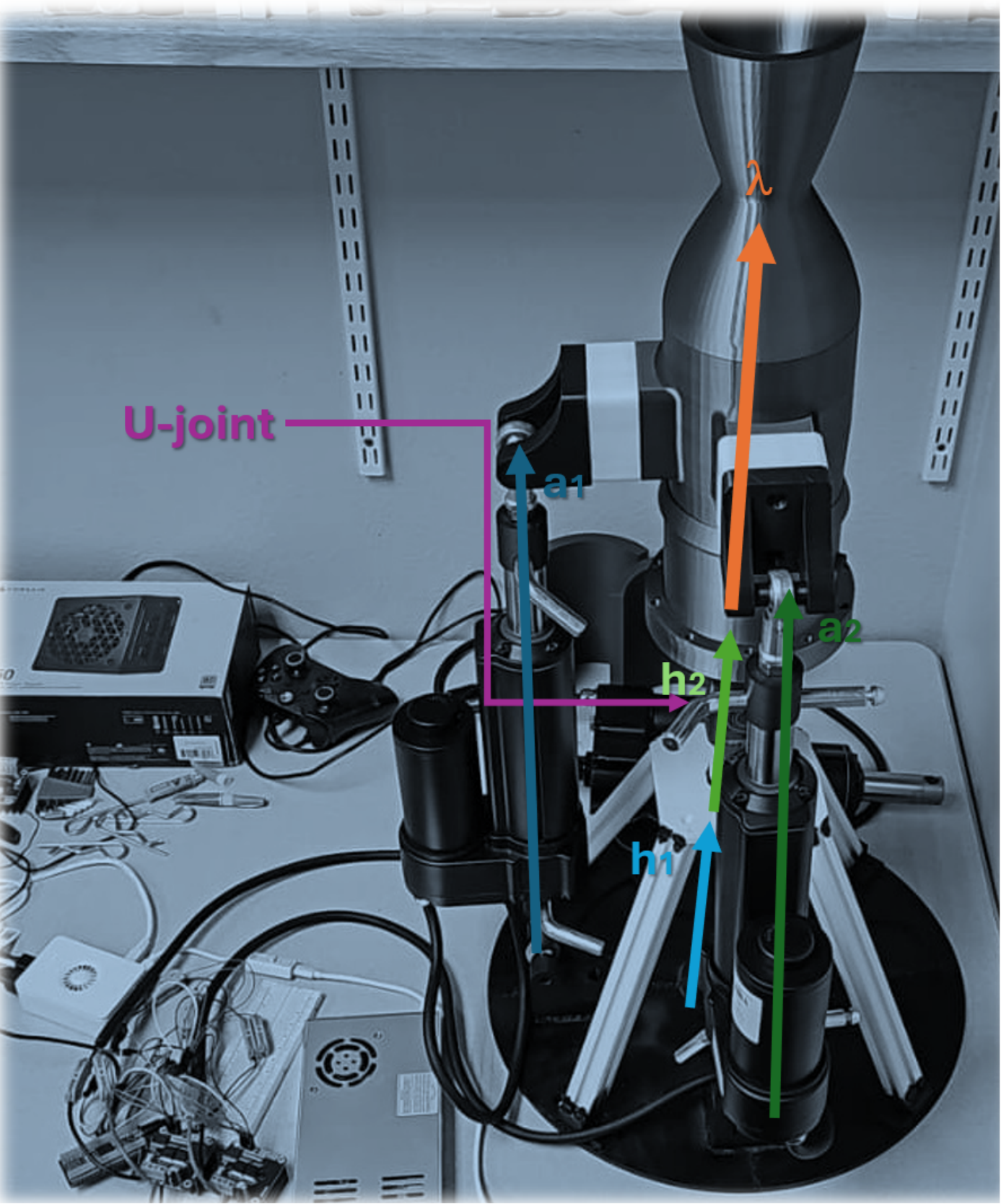}
  \caption{Experimental TVC platform. Actuator
  vectors $\bm{a}_1$, $\bm{a}_2$ connect base pins to top-mount
  pins through the universal joint; $\lam$ is the thrust direction.}
  \label{fig:platform}
\end{figure}

The orientation of the top plate is described by the rotation
matrix $R = [\,\hat{\bm r}_1 \;\; \hat{\bm r}_2, \;\; \hat{\bm r}_3\,] \in SO(3)$,
where the thrust direction is
\begin{equation}\label{eq:lambda}
  \lam \;=\; \hat{\bm r}_3 \;\in\; \Sph\,.
\end{equation}
Since the universal joint blocks axial rotation, the full
rotation~$R$ is uniquely determined by~$\lam$.
Concretely, $R(\lam) = R_y(\theta(\lam))\,R_x(\phi(\lam))$
where
\begin{equation}\label{eq:lam2angles}
  \phi = -\!\arctan\left(\frac{\lambda_2}{\sqrt{\lambda_1^2 + \lambda_3^2}}\right)\,,
  \quad
  \theta = \arctan\left(\frac{\lambda_1}{\lambda_3}\right)\,,
\end{equation}
are intermediate algebraic quantities (not independent
coordinates) determined entirely by~$\lam$.

The actuation vectors from base pin to mount pin can be 
found from the kinematic constraint in Fig. \ref{fig:VectorEq}, 
this is,
\begin{align}
  \bm{a}_i(\lam) &= [h_1I + h_2\,R(\lam)]\hat{\bm{e}}_3
    + [r_t\,R(\lam) - r_bI]\,\hat{\bm{e}}_i~|_{i=1,~2}\,,
    \label{eq:bvec}
\end{align}
and the actuator extensions can be found from
$q_i + \ell = \norm{\bm{a}_i}$, for $i=1,~2$,
where $\ell$ is the actuator rest length.

\begin{figure}[htb!]
    \centering
    \includegraphics[width=0.95\linewidth]{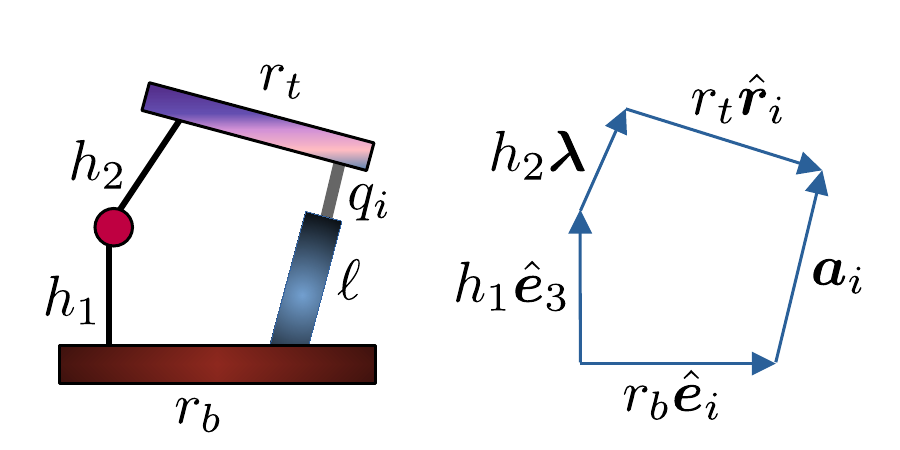}
    \caption{Kinematic constraint of the parallel mechanism.
    Actuator vectors $\bm{a}_1$, $\bm{a}_2$ form a closed
    chain through the universal joint; tilting about either
    axis displaces both actuators simultaneously.}
    \label{fig:VectorEq}
\end{figure}

\begin{mydef}[Thrust-to-actuator map]
The map $\q : \Sph \to \R^2$ is defined by
\begin{equation}\label{Eq:map}
    \q(\lam) = [\,\norm{\bm{a}_1(\lam)} - \ell,\;
              \norm{\bm{a}_2(\lam)} - \ell\,]\tp
\end{equation}
\end{mydef}

\begin{mydef}[Actuator Jacobian]
The Jacobian of the thrust-to-actuator map, restricted to the
tangent space $T_{\lam}\Sph$, is
\begin{equation}\label{eq:jac}
  \mathcal{J}(\lam) \;=\;
  \frac{\partial \q}{\partial \lam}\bigg|_{T_{\lam}\Sph}
  \;\in\; \R^{2 \times 2}\,.
\end{equation}
\end{mydef}

\begin{proposition}[Local invertibility of $\mathcal{J}$]\label{prop:jinv}
The Jacobian $\mathcal{J}(\lam)$ is locally invertible in a neighborhood
of the nominal configuration $\lam_0 = \hat{\bm{e}}_3$, where
$\det(\mathcal{J})\big|_{\lam_0} = r_b\,r_t > 0$.
\end{proposition}
\begin{proof}
The map $\lam \mapsto \det(\mathcal{J}(\lam))$ is continuous.
Since $\det(\mathcal{J})|_{\lam_0} = r_b\,r_t > 0$,
there exists an open ball $\mathcal{O}(\lam_0) \subset \Sph$,
centered at $\lam_0$, such that
$\det(\mathcal{J}(\lam)) > 0$ for all
$\lam \in \mathcal{O}(\lam_0)$.
\end{proof}

\subsection{Actuator-Space Dynamics}

The coupled equation of motion projected onto actuator
space takes the standard Euler--Lagrange form
\begin{equation}\label{eq:eom}
  M_q(\q)\,\ddot{\q} + C_q(\q,\dot{\q})\,\dot{\q}
    + \bm{g}_q = \bm{u} + \bm{d}\,,
\end{equation}
where $M_q \in \R^{2\times 2}$ is symmetric, positive
definite, and dominated by diagonal terms due to the
gearbox-decoupled actuator inertias; $C_q$ satisfies
the skew-symmetry property $\dot{M}_q = C_q + C_q\tp$;
$\bm{g}_q = m_\mathrm{act}\,g\,[1,\,1]\tp$ is the
gravitational load, with $m_\mathrm{act}$ the actuator
mass and $g$ the gravitational acceleration;
$\bm{u} = [u_1,\,u_2]\tp$ the actuator forces; and
$\bm{d}$ condenses unmodeled effects such as backlash,
friction, and cross-coupling. The derivation from first
principles is omitted, as the controller developed in
Section~III requires only the well-known structural properties
of~\eqref{eq:eom}, not the explicit expressions for
$M_q$ or $C_q$.

\section{PIDNet Controller}\label{sec:controller}

\subsection{Control Objective}

Track a desired thrust direction $\lamdes(t) \in \Sph$
using actuator forces $\bm{u} \in \R^2$. The reference
is mapped to actuator space as
$\qdes = \q(\lamdes(t))$ from~\eqref{Eq:map}, where
local invertibility is guaranteed by
Proposition~\ref{prop:jinv}.

\subsection{PD Extended Error}

Define the tracking error and the
PD extended error
\begin{align}\label{eq:surface}
    \qtil &= \q - \qdes,\\
      \s  &= \dot{\qtil} + \alpha\,\qtil\,,
\end{align}
where $\alpha > 0$ is the desired convergence rate. When
$\s = \bm{0}$, the error satisfies
$\dot{\qtil} + \alpha\,\qtil = \bm{0}$, which is
exponentially stable.

\subsection{Error Dynamics}

Substituting $\ddot{\q} = \dot{\s} - \alpha\,\dot{\qtil}
+ \ddot{\q}_{\mathrm{des}}$ into~\eqref{eq:eom} yields
\begin{equation}\label{eq:errdyn}
  M_q\,\dot{\s} + C_q\,\s
  \;=\; \bm{u} + \bm{d} - \bm{g}_q + \bm{Y}_r\,,
\end{equation}
where $\bm{Y}_r$ collects inertial and Coriolis terms
evaluated along the reference trajectory.

The inertia matrix allows $M_q(\q) = \bar{m}\,I + \Delta M_q(\q)$, where
$\bar{m}$ is the dominant diagonal entry set by the
gearbox-reflected actuator inertia. Because the
actuator gearboxes contribute inertias that are orders
of magnitude larger than those of the upper mount,
$\norm{\Delta M_q}$ is small relative to $\bar{m}$
over the compact operating region. Then
\begin{equation}
  M_q\,\dot{\s}
  \;=\; \bar{m}\,\dot{\s} + \boldsymbol{\varpi}\,,
  \qquad
  \boldsymbol{\varpi} = \Delta M_q(\q)\,\dot{\s}\,,
\end{equation}
where $\boldsymbol{\varpi}$ is bounded because
$\Delta M_q$ is bounded on the compact set and
$\dot{\s}$ is physically bounded by the actuator
force and velocity limits.

\begin{remark}
    Rather than relying on structural assumptions about the
    plant model, the proposed controller is designed to
    command hardware in the real-world, directly in actuator space, and without
    requiring explicit knowledge of the inertia, Coriolis, or
    gravitational terms. The mathematical guarantees developed
    in Section~\ref{sec:stability} are a consequence of the
    system mechanics, such as bounded actuation, compact operating
    region, and the passivity structure inherent to
    Euler--Lagrange systems, and these are not prerequisites
    imposed on a known or idealized model.
\end{remark}

Grouping all the dynamical terms that the controller
does not compute online, and invoking the universal
approximation property of RBF
networks~\cite{lewis1998} for the basis functions
$\bm{\Phi}(\qtil,\dot{\qtil}) \in \R^{p \times 2}$
defined in~\eqref{eq:rbf}, there exist ideal weights
$\bm{\beta}^* \in \R^p$ such that
\begin{equation}\label{eq:approx}
  \bm{d} - \bm{g}_q + \bm{Y}_r
  - C_q\,\s - \boldsymbol{\varpi}
  \;=\; \bm{\Phi}\tp\bm{\beta}^*
  + \bm{\varepsilon}\,,
\end{equation}
where $\bm{\varepsilon}$ is the approximation residual,
bounded in the sense that
$\limsup_{t\to\infty}\norm{\bm{\varepsilon}(t)}
\leq \bar{\varepsilon} < \infty$,
as $\bm{\varepsilon}$ is nonzero only when the
supervisory gate~\eqref{eq:superv} is active
($|s_i| < s_c$). The closed-loop model then reads
\begin{equation}\label{eq:app_errdyn}
  \bar{m}\,\dot{\s}
  \;=\; \bm{u} + \bm{\Phi}\tp\bm{\beta}^*
  + \bm{\varepsilon}\,.
\end{equation}
\subsection{Control Law}

The PIDNet actuator force is composed of a bounded
nonlinear PD action and a neural integral term:
\begin{equation}\label{eq:control}
  \bm{u} \;=\; -k_d\,\boldsymbol{\Psi}(\s)
    \;-\; \bm{\Phi}\tp\bm{\beta}\,,
\end{equation}
where $\boldsymbol{\Psi}(\s) = \tanh(\nu\,\s)$ is applied
element-wise;
$k_d > 0$ is the nonlinear PD gain, and
$\nu > 0$ controls the linear amplification near the
origin. The weight vector $\bm{\beta} \in \R^p$ evolves
according to
\begin{equation}\label{eq:adapt}
  \dot{\bm{\beta}}
  \;=\; \bm{\Gamma}\,\bm{\Phi}\,\boldsymbol{\Psi}(\s)\,,
\end{equation}
with $\bm{\Gamma} = \mathrm{diag}(\gamma_1, \ldots,
\gamma_p) \succ 0$ the learning rate matrix. Integrating
\eqref{eq:adapt} and substituting into \eqref{eq:control}
reveals the PID structure:
\begin{equation}\label{eq:pid_structure}
  \bm{u} \;=\;
  \underbrace{-k_d\,\boldsymbol{\Psi}(\s)}_{\text{nonlinear PD}}
  \;\underbrace{-\;
  \bm{\Phi}\tp\bm{\Gamma}
  \int_0^t \bm{\Phi}\,\boldsymbol{\Psi}(\s)\,d\tau
  }_{\text{neural integral action}}\,,
\end{equation}
where the product $\bm{\Phi}\tp\bm{\Gamma}$ plays the
role of a state-dependent integral gain: unlike a
constant $K_i$, it modulates the accumulation rate
through the regressor $\bm{\Phi}$, selectively
integrating in regions of the phase plane where the
corresponding basis functions are active.

\begin{remark}
In the presence of persistent approximation error, a
sigma-modification term
$-\gamma_0\,\bm{\Gamma}\,\bm{\beta}$ may be appended
to~\eqref{eq:adapt} to prevent weight drift, at the
cost of a nonzero steady-state bound. 
\end{remark}

\subsection{RBF Basis Functions}

The regressor $\bm{\Phi}(\qtil,\dot{\qtil}) \in
\R^{p \times 2}$ is constructed entirely from the
tracking error and its time derivative, without any
structural knowledge of the plant dynamics. Unlike
model-based regressors that require explicit
computation of the inertia, Coriolis, and gravitational
matrices, $\bm{\Phi}$ depends only on the phase-plane
coordinates
$\bm{\xi}_i = [\tilde{q}_i,\;\dot{\tilde{q}}_i]\tp$
through a constant bias and Gaussian kernels, gated by
a supervisory switch:
\begin{equation}\label{eq:rbf}
  \bm{\Phi}
  =
  \begin{bmatrix}
    1 & 1 \\
    \varphi_{0,1} & \varphi_{0,2} \\
    \vdots & \vdots \\
    \varphi_{4,1} & \varphi_{4,2}
  \end{bmatrix}
  S_v\,,
  \quad
  \varphi_{k,i}
  = \exp\!\left(
    -\frac{\norm{\bm{\xi}_i - \bm{c}_k}^2}{2\sigma^2}
  \right)\!,
\end{equation}
where $\bm{c}_0=\bm{0}$ captures near-equilibrium effects,
$\bm{c}_{1\text{--}4}$ cover the four sign quadrants, and the bias row
ensures a nonzero regressor. The diagonal gate
\begin{equation}\label{eq:superv}
  S_v = \mathrm{diag}(s_{v,1},\; s_{v,2})\,,
  \quad
  s_{v,i} =
  \begin{cases}
    1 & |s_i| < s_c\,,\\
    0 & \text{otherwise}\,,
  \end{cases}
\end{equation}
restricts the regressor and weight updates to the neighborhood of the
error manifold.
\subsection{Adaptation Law}

In practice, the approximation error $\bm{\varepsilon}$
in~\eqref{eq:approx} is nonzero, and the ideal integral
law~\eqref{eq:adapt} may exhibit weight drift. Thus, to ensure
bounded parameter evolution, a sigma-modification term
is appended:
\begin{equation}\label{eq:adapt_sigma}
  \dot{\bm{\beta}}
  \;=\; \bm{\Gamma}\,\bm{\Phi}\,\boldsymbol{\Psi}(\s)
        \;-\; \gamma_0\,\bm{\Gamma}\,\bm{\beta}\,,
\end{equation}
where $\gamma_0 > 0$ is the sigma-modification gain.
The first term drives the adaptation in the direction that
reduces the extended error, while the second provides a
restoring force toward the origin in the weight-vector space,
preventing unbounded growth at the cost of a nonzero
residual bound quantified in Section~IV-B.

The diagonal matrix $\bm{\Gamma}$ assigns independent learning rates to
each basis function, where higher rates are assigned to same-sign error-velocity
quadrants, and lower rates are used for opposite-sign error-velocity
quadrants.

\begin{remark}
The control law~\eqref{eq:control} with
adaptation~\eqref{eq:adapt_sigma} requires only
actuator-displacement measurements $\q$ and the precomputed
reference $\qdes(\lamdes)$. Neither the
Jacobian~$\mathcal{J}$, the inertia~$M_q$, nor the
Coriolis matrix~$C_q$ are computed online.
\end{remark}

\section{Stability Analysis}\label{sec:stability}

The results in this section are established for the
closed-loop model~\eqref{eq:app_errdyn}, which retains
the dominant actuator inertia and absorbs all unmodeled
dynamics: friction, backlash, cross-coupling, and
off-diagonal inertia residuals into the RBF
approximation target. No explicit expressions for
$M_q$, $C_q$, or $\bm{g}_q$ are required at any point
in the analysis.

\subsection{Ideal Case: $\bar{\varepsilon} = 0$, $\gamma_0 = 0$}

Consider first the ideal adaptation law~\eqref{eq:adapt} with perfect
approximation ($\bm{\varepsilon} = \bm{0}$). Define the composite
storage function
\begin{equation}\label{eq:lyap}
  V \;=\;
  \frac{\bar{m}}{\nu}\sum_{i=1}^{2}
    \log\cosh(\nu\,s_i)
  \;+\;
  \frac{1}{2}\,\betat\tp\bm{\Gamma}^{-1}\betat\,,
\end{equation}
where $\betat = \bm{\beta} - \bm{\beta}^*$ is the weight estimation
error. Since $\log\cosh(\cdot) \geq 0$ with equality only at the origin,
and $\bm{\Gamma}^{-1} \succ 0$, one has $V > 0$ for
$(\s,\betat) \neq \bm{0}$ and $V = 0$ only at the origin.

Differentiating~\eqref{eq:lyap} and using
$\tfrac{d}{ds}\log\cosh(\nu s) = \nu\tanh(\nu s)$ yields
\begin{equation}\label{eq:vdot1}
  \dot{V}
  \;=\; \bar{m}\,\boldsymbol{\Psi}\tp\dot{\s}
  \;+\; \betat\tp\bm{\Gamma}^{-1}\dot{\bm{\beta}}\,.
\end{equation}
Now, substituting the error dynamics~\eqref{eq:app_errdyn} with
control~\eqref{eq:control} into the first term gives
\begin{equation}\label{eq:vdot_expand}
  \bar{m}\,\boldsymbol{\Psi}\tp\dot{\s}
  \;=\;
  -k_d\norm{\boldsymbol{\Psi}}^2
  - \boldsymbol{\Psi}\tp\bm{\Phi}\tp\betat
  + \boldsymbol{\Psi}\tp\bm{\varepsilon}\,.
\end{equation}
In addition, substituting the ideal adaptation law~\eqref{eq:adapt} into the second
term produces
\begin{equation}\label{eq:vdot_adapt}
  \betat\tp\bm{\Gamma}^{-1}\dot{\bm{\beta}}
  \;=\;
  \betat\tp\bm{\Phi}\,\boldsymbol{\Psi}
  \;=\;
  \boldsymbol{\Psi}\tp\bm{\Phi}\tp\betat\,.
\end{equation}
Finally, summing~\eqref{eq:vdot_expand} and~\eqref{eq:vdot_adapt}, the cross
terms $\mp\boldsymbol{\Psi}\tp\bm{\Phi}\tp\betat$ cancel exactly,
leaving
\begin{equation}\label{eq:vdot_ideal}
  \dot{V}
  \;=\; -k_d\norm{\boldsymbol{\Psi}}^2
    \;+\; \boldsymbol{\Psi}\tp\bm{\varepsilon}\,.
\end{equation}

\begin{proposition}[Asymptotic stability]\label{prop:asymptotic}
Consider the error dynamics~\eqref{eq:app_errdyn} under control
law~\eqref{eq:control} with ideal adaptation~\eqref{eq:adapt}. If
$\limsup_{t\to\infty}\norm{\bm{\varepsilon}(t)} = 0$, then
$\limsup_{t\to\infty}\norm{\s(t)} = 0$ and
$\limsup_{t\to\infty}\norm{\qtil(t)} = 0$.
\end{proposition}
\begin{proof}
From~\eqref{eq:vdot_ideal}, $\dot{V} \leq 0$ whenever
$\norm{\boldsymbol{\Psi}} \geq
\frac{1}{k_d}\norm{\bm{\varepsilon}}$.
Hence $\norm{\boldsymbol{\Psi}}$ converges in finite
time to a ball of radius
$\frac{1}{k_d}\norm{\bm{\varepsilon}}$ centered at the
origin. The result then follows from
$\limsup_{t\to\infty}\norm{\bm{\varepsilon}(t)} = 0$.
\end{proof}

\subsection{Practical Case: $\bar{\varepsilon} > 0$, $\gamma_0 > 0$}

When the approximation residual is persistently nonzero,
the ideal law~\eqref{eq:adapt} may produce unbounded
weight growth. The sigma-modified
law~\eqref{eq:adapt_sigma} prevents this by introducing
a restoring term that penalizes large weights.
To quantify the resulting bound, note that
$\norm{\boldsymbol{\Psi}(\s)} \leq \sqrt{2}$ by the
saturation of $\tanh(\cdot)$, and that each column of
$\bm{\Phi}$ satisfies
$\norm{\bm{\Phi}_{\cdot,j}} \leq \sqrt{p}$, since the
Gaussian kernels, bias, and gate are each bounded by
unity. Together, these imply a uniform bound on the
driving term in~\eqref{eq:adapt_sigma} as
$\norm{\bm{\Gamma}\bm{\Phi}\boldsymbol{\Psi}} \leq
\gamma_{\max}\sqrt{2p}$.

\begin{theorem}[Weight boundedness]\label{thm:beta_bound}
Under the sigma-modified adaptation~\eqref{eq:adapt_sigma} with
$\gamma_0 > 0$, the weight vector $\bm{\beta}(t)$ is uniformly
ultimately bounded, this is,
$\limsup_{t\to\infty}\norm{\bm{\beta}(t)} \leq \beta_{\max}$,
where
\begin{equation}\label{eq:beta_bound}
  \beta_{\max}
  \;=\;
  \frac{\gamma_{\max}\sqrt{2p}}{\gamma_0\,\gamma_{\min}}\,,
  \qquad
  \gamma_{\max} = \max_i \gamma_i\,,
  \quad
  \gamma_{\min} = \min_i \gamma_i\,
\end{equation}
where $\gamma_i$ is the $i$-th diagonal element of the matrix $\Gamma$.
\end{theorem}
\begin{proof}
Differentiating $\norm{\bm{\beta}}$ along~\eqref{eq:adapt_sigma}
and applying the Cauchy--Schwarz inequality to the
first term gives
\begin{align*}
    \frac{d}{dt} \norm{\bm{\beta}}
    &= \frac{\bm{\beta}\tp}{\norm{\bm{\beta}}}
       \bm{\Gamma}\,\bm{\Phi}\,\boldsymbol{\Psi}(\s)
       \;-\; \gamma_0\,
       \frac{\bm{\beta}\tp}{\norm{\bm{\beta}}}
       \bm{\Gamma}\,\bm{\beta}\\[4pt]
    &\leq \gamma_{\max}\sqrt{2p}
       \;-\; \gamma_0\,\gamma_{\min}\,
       \norm{\bm{\beta}}\,,
\end{align*}
where the upper bound uses
$\norm{\bm{\Gamma}\bm{\Phi}\boldsymbol{\Psi}}
\leq \gamma_{\max}\sqrt{2p}$ and the lower bound uses
$\bm{\beta}\tp\bm{\Gamma}\bm{\beta} \geq
\gamma_{\min}\norm{\bm{\beta}}^2$.
Therefore $\frac{d}{dt}\norm{\bm{\beta}} \leq 0$
whenever $\norm{\bm{\beta}} \geq \beta_{\max}$.
\end{proof}

\begin{remark}[Small-perturbation interpretation]
Proposition~\ref{prop:asymptotic} establishes that the
ideal closed loop ($\bar{\varepsilon}=0$, $\gamma_0=0$)
drives $\s(t)\to\bm{0}$.
Theorem~\ref{thm:beta_bound} guarantees that the
sigma-modified weights remain bounded for any
$\gamma_0>0$. Together, these two results imply that the difference between
the ideal and the practical closed-loop systems is a bounded perturbation.
Because the vector fields are uniformly
bounded and locally Lipschitz on the compact operating
region, continuous dependence of solutions guarantees
that trajectories under
$(\bar{\varepsilon},\gamma_0)\to(0,0)$ converge
uniformly on every finite interval $[0,T]$ to the ideal
trajectory. In particular, for any tolerance $\rho>0$
and any finite horizon, sufficiently small
$\bar{\varepsilon}$ and $\gamma_0$ confine the
practical extended error to a $\rho$-neighborhood of
the ideal solution. This connects the two preceding
results: the ideal analysis provides the asymptotic
target, and the weight bound ensures that the practical
system remains close to it.
\end{remark}

\begin{remark}[Practical PD]
When $\gamma_0 \gg 1$, the sigma-modification term
dominates the adaptation law~\eqref{eq:adapt_sigma},
driving the weights toward the origin after a fast
transient with time constant $1/\gamma_0$. In this
regime, $\bm{\beta} \approx \bm{0}$ and the PIDNet
controller reduces to the nonlinear PD action
$\bm{u} \approx -k_d\,\boldsymbol{\Psi}(\s)$.
This limiting case is realized experimentally as the
PIDNet~(off) configuration in
Section~\ref{sec:experiments}, isolating the
contribution of the bounded nonlinear PD from the
adaptive layer and providing the non-adaptive baseline.
\end{remark}

\subsection{$\mathcal{H}_\infty$ Gain Bound}

Performance guarantees can be derived under either the
ideal adaptation law~\eqref{eq:adapt} ($\gamma_0 = 0$)
or the practical nonlinear PD of Remark~2
($\gamma_0 \gg 1$). In either case, the dissipation
inequality~\eqref{eq:vdot_ideal} applies, giving
\begin{equation}\label{eq:vdot_hinf_psi}
  \dot{V}
  =
  -k_d\norm{\boldsymbol{\Psi}}^2
  +
  \boldsymbol{\Psi}\tp\bm{\varepsilon};
\end{equation}
however, here the residual $\bm{\varepsilon}$ plays a different role in each
case: under ideal adaptation ($\gamma_0 = 0$), it is the
RBF approximation residual; under the practical
nonlinear PD ($\gamma_0 \gg 1$,
$\bm{\beta} \approx \bm{0}$), it absorbs the full
uncompensated dynamics
$\bm{d} - \bm{g}_q + \bm{Y}_r - C_q\,\s
+ \boldsymbol{\varpi}$. In both cases,
$\limsup_{t\to\infty}\norm{\bm{\varepsilon}(t)}
\leq \bar{\varepsilon} < \infty$ over the compact
operating region.

Applying Young's inequality with free parameter
$\kappa > 0$, one gets
\begin{equation}\label{eq:young_psi}
  \boldsymbol{\Psi}\tp\bm{\varepsilon}
  \leq
  \frac{1}{2\kappa^2}\norm{\boldsymbol{\Psi}}^2
  +
  \frac{\kappa^2}{2}\norm{\bm{\varepsilon}}^2,
\end{equation}
which yields
\begin{equation}\label{eq:vdot_psi_bound}
  \dot{V}
  \leq
  -\left(k_d - \frac{1}{2\kappa^2}\right)
  \norm{\boldsymbol{\Psi}}^2
  +
  \frac{\kappa^2}{2}\norm{\bm{\varepsilon}}^2.
\end{equation}
Moreover, provided
\begin{equation}\label{eq:kd_condition}
  c = k_d - \frac{1}{2\kappa^2} > 0,
\end{equation}
integrating from $t_1$ to $t_2$ and using
$V(t_2) \geq 0$ gives
\begin{equation}\label{eq:psi_l2_bound}
  c\int_{t_1}^{t_2}\norm{\boldsymbol{\Psi}(t)}^2\,dt
  \leq
  V(t_1)
  +
  \frac{\kappa^2}{2}
  \int_{t_1}^{t_2}\norm{\bm{\varepsilon}(t)}^2\,dt.
\end{equation}

\begin{theorem}[$\mathcal{H}_\infty$ attenuation of the
saturated extended error]
\label{thm:hinf}
Under the control law~\eqref{eq:control} with either
ideal adaptation~\eqref{eq:adapt} ($\gamma_0=0$) or the
practical nonlinear PD ($\gamma_0\gg1$), if
$k_d > 1/(2\kappa^2)$, then the closed-loop system
achieves finite $\mathcal{L}_2$ gain from the residual
$\bm{\varepsilon}$ to the saturated extended error
$\boldsymbol{\Psi}(\s) = \tanh(\nu\s)$, with gain no
greater than $\kappa/\sqrt{2c}$, where
$c = k_d - 1/(2\kappa^2)$.
\end{theorem}

\begin{remark}[From $\boldsymbol{\Psi}$ to $\qtil$]
Theorem~\ref{thm:hinf} bounds the $\mathcal{L}_2$ gain
to $\boldsymbol{\Psi}(\s)$, not directly to the
tracking error $\qtil$. The two are connected as
follows. First, in the non-saturated region
$|\Psi_i| \leq \bar{\psi} < 1$, the inverse relation
$s_i = \frac{1}{\nu}\operatorname{artanh}(\Psi_i)$
ensures that $\s$ remains uniformly bounded.
Second, the actuator tracking error satisfies
$\dot{\qtil} + \alpha\qtil = \s$, so $\qtil$ is
obtained by filtering $\s$ through a stable first-order
system with gain $1/\alpha$. Finally, in the local linearized
regime, where
$\boldsymbol{\Psi}(\s) \approx \nu\s$, composing these
bounds gives an end-to-end $\mathcal{L}_2$ gain from
$\bm{\varepsilon}$ to $\qtil$, as
\begin{equation}\label{eq:q_gain_from_eps}
  \gamma_{\varepsilon \to q}
  \leq
  \frac{\kappa}{\sqrt{2c}\,\alpha\nu}\,.
\end{equation}
This quantity is the tunable performance certificate of
the controller, where increasing $k_d$ tightens the
attenuation, while $\alpha$ and $\nu$ set the
convergence rate and linear amplification range,
respectively. 
A nonzero
$\qtil(t_1)$ contributes an additional transient, which is
independent of $\bm{\varepsilon}$ and vanishes
exponentially at rate $\alpha$.
\end{remark}

\begin{remark}[Supervisory threshold selection]
The threshold $s_c$ in~\eqref{eq:superv} should be
chosen so that the adaptation remains active throughout
the non-saturated region where
Theorem~\ref{thm:hinf} applies. Selecting $s_c > 1/\nu$
ensures that the RBF update is active over the range
where $\tanh(\nu s_i)$ transitions from linear behavior
toward saturation.
\end{remark}

\begin{remark}[Extension to thrust-direction space]
Since $\delta\q \approx \mathcal{J}\,\delta\lam$ in the
neighborhood of Proposition~\ref{prop:jinv}, the
actuator-space tracking bound extends to the thrust
direction as
$$\int\norm{\lamtil}^2\,dt \lesssim
\sigma_{\min}^{-2}(\mathcal{J})
\int\norm{\qtil}^2\,dt\,,$$
where $\sigma_{\min}(\mathcal{J})$ is the minimum
singular value of the actuator Jacobian. Thus, reducing
the sustained actuator error directly improves pointing
accuracy.
\end{remark}

\section{Experimental Results}\label{sec:experiments}

\subsection{Hardware Setup}

The experimental platform consists of two PA-HD2-4-2000-HS linear
actuators with quadrature encoders, driven by BTS7960 H-bridges and
controlled at 500~Hz through a Teensy~4.1/Raspberry~Pi~4 architecture.
The geometry is $h_1=33.7$\,cm, $h_2=24.2$\,cm, and
$r_b=r_t=15.7$\,cm. A Levant differentiator~\cite{levant2003} estimates
velocity from encoder position.
Safety limits enforce $\pm 5.08$\,cm stroke and $45^\circ$ total tilt.
The supervisory threshold is set to
$s_c = 12.7$\,cm\,s$^{-1}$, satisfying the condition
$s_c > 1/\nu$ from Section~\ref{sec:stability}.

\begin{remark}
The scalar $\bar{m}$ appears in the Lyapunov
function~\eqref{eq:lyap} as an analysis parameter; it
does not enter the control law~\eqref{eq:control} or the
adaptation law~\eqref{eq:adapt_sigma} and is never
computed or estimated online.
\end{remark}

\subsection{Startup Ramp}

To avoid an impulsive control action at startup, the
effective error is ramped as
$\qtil_{\mathrm{eff}}(t) = \tanh^2(\vartheta t)\,\qtil(t)$,
with $\vartheta > 0$ large enough that
$\qtil_{\mathrm{eff}}$ and its first derivative both
vanish at $t = 0$ and recover the true error within a
few seconds.

\subsection{Reference Trajectory}

The desired thrust vector traces a circle on the unit
sphere:
\begin{equation}\label{eq:ref}
  \lamdes(t) =
  \begin{bmatrix}
    \sin\omega_0\cos(\omega_1 t) \\
    \sin\omega_0\sin(\omega_1 t) \\
    \cos\omega_0
  \end{bmatrix},
\end{equation}
with $\omega_0 = 10^\circ$ and
$\omega_1 = 0.2$\,rad/s, over a 60\,s horizon. All
controllers are evaluated under identical conditions
from the same home position.

\subsection{Controller Comparison}

Four controllers are compared: a PID
with anti-windup, a super-twisting algorithm (STA),
the proposed PIDNet with adaptation disabled
($\bm{\Gamma} = \bm{0}$), and the PIDNet with
online adaptation. The PIDNet parameters are:
nonlinear PD gain $k_d = 0.9$, convergence rate
$\alpha = 10$, linear amplification
$\nu = 2.31\,\text{s\,cm}^{-1}$,
sigma-modification $\gamma_0 = 40$, learning rates
$\bm{\Gamma} = \mathrm{diag}(10^{-5},\; 0.5,\; 0.2,\;
0.07,\; 0.2,\; 0.07)$, RBF centers at
$\bm{c}_0 = \bm{0}$ and
$\bm{c}_{1\text{--}4} = (\pm 0.51,\,\pm 0.51)$\,cm
with kernel width $\sigma = 0.51$\,cm. Same-sign
quadrants receive higher learning rates to capture
direction-dependent stiction compensation. The bias
channel learning rate is set to $10^{-5}$ to suppress
drift from persistent activation.

The PID gains ($K_p = 5{\times}10^{-3}$,
$K_i = 2{\times}10^{-4}$,
$K_d = 2{\times}10^{-4}$) were tuned by manual
iteration, as classical rules are inapplicable to the nonlinear
coupled plant; anti-windup at $K_a = 0.5 K_i$ was added to mitigate
integrator saturation.
The STA gains
($k_1 = 0.02$, $k_2 = 0.015$, $\alpha = 10$) were
selected to minimize chattering while maintaining
convergence. In both cases, tuning demanded multiple
experimental runs with no systematic procedure. By
contrast, PIDNet stability requires only
$k_d > 1/(2\kappa^2)$
from~\eqref{eq:kd_condition}, with the learning rates
$\bm{\Gamma}$ affecting only adaptation speed, not
stability. The PID and STA tuning
process resulted in multiple driver failures due to
sustained actuator saturation, underscoring the
practical value of the bounded control structure
in~\eqref{eq:control}.

\subsection{Results}

Performance is quantified by the following indices:
\begin{align*}
  \mathrm{ISE}  &= \textstyle\int_0^T\!\norm{\qtil}^2\,d\tau\,,\\
  \mathrm{ITNE} &= \textstyle\int_0^T\! t\,\norm{\qtil}\,d\tau\,,\\
  \mathrm{ISC}  &= \textstyle\int_0^T\!\norm{\bm{u}}^2\,d\tau\,.
\end{align*}

The steady-state tracking performance is characterized by
$\bar{e}_{\mathrm{ss}} = \frac{1}{T-t_0}\int_{t_0}^{T}\norm{\qtil}\,d\tau$
and
$e_{\max} = \max_{t \geq t_0}\norm{\qtil(t)}$,
with $t_0 = 20$\,s to exclude the initial transient.

\begin{table}[t]
\centering
\caption{Performance Index Comparison (Circle, 60\,s). Steady-state
metrics $\bar{e}_{\mathrm{ss}}$ and $e_{\max}$ computed for
$t \geq 20$\,s.}
\label{tab:results}
\begin{tabular}{lccccc}
\hline
\textbf{Controller} & \textbf{ISE} & \textbf{ITNE}
  & \textbf{ISC} & $\bar{e}_{\mathrm{ss}}$ & $e_{\max}$ \\
  & [cm$^2$s] & [cm\,s$^2$] & [pwm$^2$s] & [cm] & [cm] \\
\hline
PID              & 3.27  & 212.1 & 55.5 & 0.108 & 0.121 \\
STA              & 2.11  & 189.3 & 52.0 & 0.104 & 0.113 \\
PIDNet (off)     & 1.82  & 126.0 & 51.5 & 0.068 & 0.080 \\
PIDNet (on)      & 1.66  & 113.0 & 51.6 & 0.061 & 0.072 \\
\hline
\end{tabular}
\end{table}

The results exhibit two distinct layers of improvement.
The first one is structural: the bounded
nonlinear PD action $\boldsymbol{\Psi}(\s)$ forces the
system onto the error manifold $\s = \bm{0}$ more
effectively than either the PID or STA baselines,
reducing the ITNE by 47\% and 40\%, respectively, with
no increase in control effort.

The second layer is the adaptive compensation. The
supervisory gate~\eqref{eq:superv} activates the RBF
network only when $|s_i| < s_c$, restricting adaptation
to the neighborhood of the error manifold where the
basis functions provide valid approximation. 
The network selectively accumulates corrections
for stiction, backlash, and direction-dependent friction
that a fixed-gain controller cannot resolve, reducing
the ITNE by a further 10\% and the worst-case
steady-state error from 0.080\,cm to 0.072\,cm. The ISC
is unchanged, confirming that the improvement is
achieved through targeted compensation rather than
increased actuation effort.

Although the adaptive improvement is modest in absolute laboratory units,
it is meaningful for TVC because actuator tracking error maps through
$\mathcal{J}^{-1}$ to thrust-direction error. Thus, reducing sustained
actuator error directly improves pointing accuracy.

The PID controller saturates at the actuator limits for
14.5\% of the run (8.7\,s), predominantly during transients
and direction reversals; the STA saturates for 5.2\%
(3.1\,s). In both cases, saturation invalidates the
nominal stability analysis and requires auxiliary anti-windup
mechanisms, whose interaction with the plant dynamics is
difficult to certify. By contrast, the PIDNet control
signal never exceeds $|u_i| = 0.95$: the $\tanh$ nonlinearity
guarantees $\|\bm u\| \lesssim k_d$ by construction, so the Lyapunov
analysis holds without modification throughout the
entire operation.

Fig.~\ref{fig:beta} shows the RBF weight evolution during the
PIDNet~(on) run. The per-weight trajectories (left) confirm that the
origin-centered kernel $\beta_1$ dominates, capturing near-equilibrium
compensation, while the quadrant weights $\beta_2$--$\beta_5$ provide
symmetric direction-dependent corrections. The bias weight $\beta_0$
remains effectively frozen at the suppressed learning rate
$\gamma_1 = 10^{-5}$. The weight norm (right) remains uniformly bounded,
consistent with Theorem~\ref{thm:beta_bound}.

\begin{figure}[t]
  \centering
  \includegraphics[width=0.48\columnwidth]{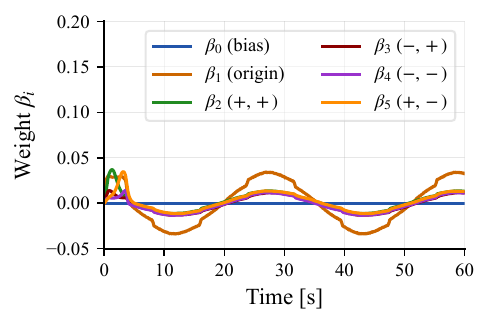}%
  \hfill
  \includegraphics[width=0.48\columnwidth]{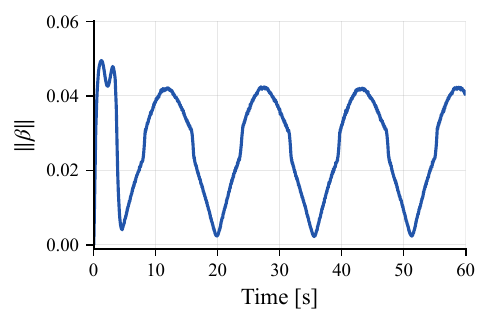}
  \caption{Weight evolution under PIDNet with adaptation.
  Left: per-weight trajectories $\beta_i(t)$.
  Right: weight norm $\norm{\bm{\beta}(t)}$.}
  \label{fig:beta}
\end{figure}

\begin{figure}[t]
  \centering
  \includegraphics[width=0.65\columnwidth]{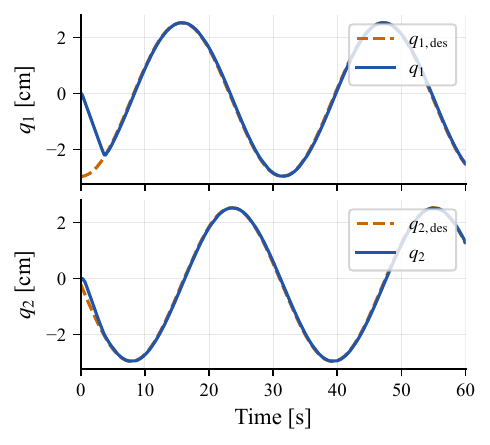}
  \caption{Actuator tracking under PIDNet with adaptation:
  reference (dashed) and measured extensions for both channels.}
  \label{fig:tracking}
\end{figure}

\begin{figure}[t]
  \centering
  \includegraphics[width=0.65\columnwidth]{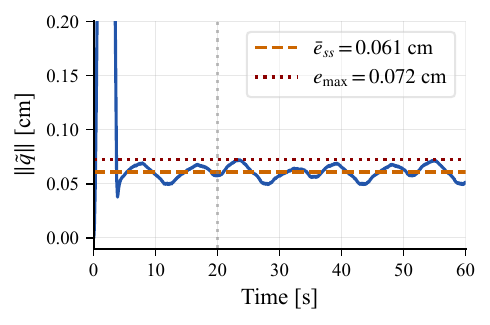}
  \caption{Tracking error norm $\norm{\qtil}$ for all four controllers.}
  \label{fig:error}
\end{figure}

\begin{figure}[t]
  \centering
  \includegraphics[width=0.65\columnwidth]{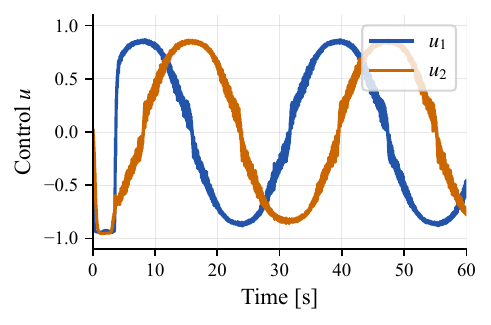}
  \caption{Control signals $u_1$, $u_2$ under PIDNet with adaptation.
  The $\tanh$ nonlinearity ensures bounded actuation throughout
  the maneuver.}
  \label{fig:control}
\end{figure}

\section{Conclusion}

This paper presented a model-free controller for a
parallel thrust vectoring mechanism that guarantees
bounded control signals by construction, eliminating the
actuator saturation events that plagued the PID and
super-twisting baselines during experimental tuning.
Stability is established without requiring explicit
knowledge of the plant dynamics: asymptotic convergence
under ideal approximation, bounded weight evolution
under persistent error via sigma-modification, and a
quantitative $\mathcal{H}_\infty$ gain bound from the
approximation residual to the actuator tracking error,
extending to the thrust-direction error through the
actuator Jacobian.

Experimental validation on a two-degree-of-freedom
electromechanical TVC rig demonstrates that the proposed
PIDNet reduces the ITNE by 47\% relative to PID and
40\% relative to a super-twisting baseline, while
tightening the worst-case steady-state error to
0.072\,cm at no additional control effort. Disabling the
adaptive layer isolates the two contributions: the
bounded nonlinear PD provides the dominant robustness
margin, while the coupled MIMO RBF network compensates
direction-dependent friction and geometric coupling that
no fixed-gain single-axis design can resolve, improving
accuracy by a further 10\%.

The controller requires only actuator-displacement
measurements and a precomputed reference map
$\qdes(\lamdes)$, with no online computation of the
Jacobian, inertia, or Coriolis matrices. No model was
identified, no regressor was constructed, and no driver
was damaged during the experimental campaign.

\bibliographystyle{asmejour}
\bibliography{Refs}




\end{document}